\documentclass[
final
]{dmtcs-episciences}

\newtheorem{theorem}{Theorem}{\bfseries}{\itshape}
\newtheorem{lemma}{Lemma}{\bfseries}{\itshape}
\newtheorem{proposition}{Proposition}{\bfseries}{\itshape}
\newtheorem{remark}{Remark}{\bfseries}{\itshape}
\newtheorem{corollary}{Corollary}{\bfseries}{\itshape}
\newtheorem{claimN}{Claim}{\bfseries}{\itshape}

\newcommand{\paraprobl}[4]
{
  \begin{flushleft}
    \fbox{
      \begin{minipage}{14.25cm}
        \noindent {\textsc {#1}}\\
        {\bf Input:} #2\\
        {\bf Parameter:} #4\\
        {\bf Question:} #3
      \end{minipage}
    }
  \end{flushleft}
}

\usepackage[utf8]{inputenc}
\usepackage{subfigure,cite,amsmath,dsfont}

\author{Luerbio Faria\affiliationmark{1}
  \and Sulamita Klein\affiliationmark{2}
  \and Ignasi Sau\affiliationmark{3,4}
  \and Rubens Sucupira\affiliationmark{1}}
\title[Improved Kernels for \textsc{Signed Max Cut ATLB} on $(r,\ell)$-graphs]{Improved Kernels for Signed Max Cut Parameterized Above Lower Bound on $(r,\ell)$-graphs\thanks{This work was partially supported by CNPq, CAPES, FAPERJ, and COFECUB.}}
\affiliation{
  Instituto de Matem\'atica e Estat\'{i}stica - UERJ and  COPPE/Sistemas - UFRJ, Rio de Janeiro, Brazil\\
  Instituto de Matem\'atica and COPPE/Sistemas - UFRJ, Rio de Janeiro, Brazil\\
  CNRS, LIRMM, Universit\'e de Montpellier, Montpellier, France\\
  Departamento de Matem\'atica, Universidade Federal do Cear\'a, Fortaleza, Brazil}

\received{2016-7-20}
\accepted{2017-5-2}
\revised{2017-3-5}

\begin{document}
\publicationdetails{19}{2017}{1}{14}{1540}
\maketitle
\begin{abstract}
  A graph $G$ is {\it signed} if each edge is assigned ``$+$'' or ``$-$''.
A signed graph is {\it balanced} if
there is a bipartition of its vertex set such that
an edge has sign ``$-$'' if and only if its endpoints are in different parts.
The Edwards-Erd\H{o}s bound states that
every signed graph with $n$ vertices and $m$ edges has a balanced subgraph with at least
$\frac{m}{2}+\frac{n-1}{4}$ edges. In the {\sc Signed Max Cut Above Tight Lower Bound}
({\sc Signed Max Cut ATLB}) problem, given a signed graph $G$ and a  parameter
$k$, the question is whether $G$ has a balanced subgraph with at least
$\frac{m}{2}+\frac{n-1}{4}+\frac{k}{4}$ edges. This problem generalizes \textsc{Max Cut Above Tight Lower Bound}, for which a kernel with $O(k^5)$ vertices was given by Crowston et al.~[ICALP 2012, Algorithmica 2015]. Crowston et al. [TCS 2013] improved this result by providing a kernel with $O(k^3)$ vertices for the more general {\sc Signed Max Cut ATLB} problem. In this article we are interested in improving the size of the kernels for {\sc Signed Max Cut ATLB} on restricted graph classes for which the problem remains hard.  For two integers $r,\ell \geq 0$, a graph $G$ is an \emph{$(r,\ell)$-graph} if $V(G)$ can be partitioned into $r$ independent
 sets and $\ell$ cliques. Building on the techniques of Crowston et al. [TCS 2013], for any $r,\ell \geq 0$ we provide a kernel with $O((r + \ell)k^2)$ vertices on $(r,\ell)$-graphs, and a simple linear kernel on subclasses of split graphs for which we prove that the problem is still {\sf{NP}}-hard.
\end{abstract}

\section{Introduction}
\label{sec:intro}

A graph $G=(V,E)$ is a {\it{signed graph}} if each edge is assigned {\it{positive}} $(``+'')$ or {\it{negative}} $(``-'')$. In this paper all graphs will be signed graphs. The labels ``$+$'' or ``$-$'' are the {\it{signs}} of the corresponding edges. If $(V_1,V_2)$ is a bipartition of the vertex set of a graph $G$, we say that $G$ is {\it{$(V_1,V_2)$-balanced}} if the edges with both endpoints in $V_1$ or both endpoints in $V_2$ are positive and the edges with endpoints in distinct parts are negative. A graph $G$ is {\it{balanced}} if there exists a partition $(V_1,V_2)$ of $V(G)$ such that $G$ is $(V_1,V_2)$-balanced. (Note that we can use an algorithm for {\sc 2-sat} to determine whether a signed graph is balanced.)
	
The problem of finding a balanced subgraph of a signed graph with maximum number of edges is called {\sc{Signed Max Cut}} and is {\sf{NP}}-hard, as it is a generalization of {\sc{Max Cut}}, which is well known to be {\sf{NP}}-hard~\cite{garey1976some} (indeed, {\sc{Max Cut}} is exactly {\sc{Signed Max Cut}} when all edges of $G$ are negative). We refer to the work of Crowston et al.~\cite{crowston2013maximum} and the references therein for some applications of the {\sc{Signed Max Cut}} problem.

In this article we study {\sc{Signed Max Cut}} from the perspective of Parameterized Complexity; we refer to the monographs~\cite{FG06,Nie06,DF13,CyganFKLMPPS15} for an introduction to the field. Namely, we consider a so-called \emph{parameterization above lower bound} of the problem, and we are particularly interested in obtaining small polynomial kernels on restricted graph classes. Let us first discuss \textsc{Max Cut}. Edwards~\cite{edwards, edwards2} showed that any connected
graph with $n$ vertices and $m$ edges\footnote{We assume that all input graphs of the problems under consideration have $n$ vertices and $m$ edges.}  has a cut of size $\frac{m}{2}+\frac{n-1}{4}$, and that this value is \emph{tight} (that is, that there exist infinitely many graphs for which the size of a maximum cut equals this bound). This bound is commonly known as the {\it{Edwards-Erd{\H{o}}s}} bound, and justifies the following parameterization of \textsc{Max Cut} recently considered by Crowston et al.~\cite{crowston2012max}: given a graph $G$ and a parameter $k$, decide whether $G$ has a cut of size at least $\frac{m}{2}+\frac{n-1}{4}+\frac{k}{4}$. They provided {\sf{FPT}}-algorithms for this parameterization of the problem and a kernel of size $O(k^5)$ on general graphs.


Coming back to the {\sc{Signed Max Cut}} problem, Poljak and Turz\'{\i}k \cite{poljak1986polynomial} proved that $\frac{m}{2}+\frac{n-t}{4}$ is also a tight lower bound on the number of edges in a balanced subgraph of any signed graph $G$ with $t$ connected components; we denote this lower bound by ${\sf pt}(G)$.
Motivated by the above result for \textsc{Max Cut}, Crowston et al. \cite{crowston2013maximum} considered the following parameterization of {\sc{Signed Max Cut}}, which is the same we consider in this article:


\vspace{.15cm}
\paraprobl
{\textsc{Signed Max Cut Above Tight Lower Bound}}
{ A connected signed graph $G$  and a positive integer $k$.}
{Does $G$ contain a balanced subgraph with at least $\frac{m}{2}+\frac{n-1}{4}+\frac{k}{4}$ edges?}
{$k$.}
\vspace{.15cm}	


We call the above problem {\sc{Signed Max Cut ATLB}} for short. Crowston et al.~\cite{crowston2013maximum} proved that {\sc{Signed Max Cut ATLB}} is also {\sf{FPT}}  and provided a kernel of size
 $O(k^3)$ on general graphs, therefore improving the kernel of size
 $O(k^5)$ for the particular case of {\sc{Max Cut ATLB}} given by Crowston et al.~\cite{crowston2012max}.

\vspace{.15cm}\noindent
\textbf{Our results}. In this article we are interested in improving the size of the cubic kernel for {\sc{Signed Max Cut ATLB}} by Crowston et al.~\cite{crowston2013maximum} on particular graph classes for which the problem remains {\sf{NP}}-hard. In particular, we focus on $(r,\ell)$-graphs. For two integers $r,\ell \geq 0$, a graph $G$ is an \emph{$(r,\ell)$-graph} if $V(G)$ can be partitioned into $r$ independent sets and $\ell$ cliques. These graph classes contain, for example, split graphs or bipartite graphs, and have been extensively studied in the literature~\cite{FoHa77,Golumbic04,book-graph-classes,Bra96,FHKM03,BFKS15,KoPa15}. In particular,
it is known that the recognition of $(r,\ell)$-graphs is {\sf NP}-complete if and only if $\max\{r,\ell\} \geq 3$~\cite{Bra96,FHKM03}.

Note that for this parameterization of  \textsc{Signed Max Cut} to make sense on $(r,\ell)$-graphs, the bound $\frac{m}{2}+\frac{n-1}{4}$ should be also {\sl tight} on these graph classes. Fortunately, this is indeed the case: let $G$ be a clique with odd number of vertices, which is clearly an $(r,\ell)$-graph as long as $\ell > 0$ (in particular, a split graph), and such that all edges are negative. It is clear that the largest balanced subgraph of $G$ is given by the crossing edges of any bipartition $(V_1,V_2)$ of $V(G)$ with $|V_1|= \frac{n+1}{2}$ and $|V_2|= \frac{n-1}{2}$. The number of edges of such a balanced subgraph, namely  $\frac{n^2-1}{4}$, equals  ${\sf pt}(G) = \frac{n(n-1)}{4}+ \frac{n-1}{4}$.

From a parameterized complexity perspective, split graphs have received considerable attention. For instance, Raman and Saurabh~\cite{RamanS08} proved that \textsc{Dominating Set} is {\sf{W[2]}}-hard on connected split graphs, and Ghosh~et al.~\cite{GhoshK0MPRR15} provided  {\sf{FPT}}-algorithms for the problem of deleting at most $k$ vertices to obtain a split graph. Concerning kernelization, Heggernes et al.~\cite{HeggernesHLS15} studied the \textsc{Disjoint Paths} problem on split graphs parameterized by the number of pairs of terminals, and provided a kernel of size $O(k^2)$ (resp. $O(k^3)$) for the vertex-disjoint (resp. edge-disjoint) version of the problem. As for general $(r,\ell)$-graphs from a parameterized point of view, recently a dichotomy on the parameterized complexity for the problem of deleting at most $k$ vertices to obtain an $(r,\ell)$-graph has been independently obtained by Baste et al.~\cite{BFKS15} and by Kolay and Panolan~\cite{KoPa15}. To the best of our knowledge, this is the first article that focuses on obtaining polynomial kernels on $(r,\ell)$-graphs for arbitrary values of $r$ and $\ell$.


Our main result (Theorem~\ref{thm:quadratic-general} in Section~\ref{sec:quadratic-kernel}) is that {\sc{Signed Max Cut ATLB}} admits a quadratic kernel on $(r,\ell)$-graphs for every fixed $r,\ell \geq 0$. More precisely, the kernel has $O((r + \ell)k^2)$ vertices. Our techniques in order to prove Theorem~\ref{thm:quadratic-general} are strongly based on the ones used by Crowston et al.~\cite{crowston2013maximum}, and for improving their cubic bound we further exploit the structure of $(r,\ell)$-graphs in the analysis of the algorithm. In fact, our kernelization algorithm consists in applying exhaustively the reduction rules of Crowston et al.~\cite{crowston2013maximum}, only the analysis changes. As when  $\max\{r,\ell\} \geq 3$ the recognition of $(r,\ell)$-graphs is {\sf NP}-complete~\cite{Bra96,FHKM03}, we stress here that, given an instance $(G,k)$, we do {\sl not} need to obtain any partition of $V(G)$ into $r$ independent sets and $\ell$ cliques; we only use the {\sl existence} of such partition in the analysis. Since the analysis of the quadratic kernel for the particular case of split graphs is simpler, we prove it separately in Subsection~\ref{ap:quadratic-split}.

In Section~\ref{sec:linear-kernel} we present a linear kernel for {\sc{Signed Max Cut ATLB}} on subclasses of split graphs. Namely, these subclasses are what we call \emph{$d^*$-split graphs} for every integer $d \geq 1$ (see Section~\ref{sec:linear-kernel} for the definition). We first prove  that even \textsc{Max Cut} is {\sf{NP}}-hard on $d^*$-split graphs for every integer $d \geq 2$, and in Theorem~\ref{thm:linear-d-split} we provide the linear kernel. As discussed later, this kernelization algorithm is the simplest possible algorithm that one could imagine, as it does {\sl nothing} to the input graph; its interest lies on the analysis of the kernel size, which is non-trivial. In particular, the analysis uses a new reduction rule that we introduce in Section~\ref{s2}.

\section{Preliminaries}
\label{s2}
	
We use standard graph-theoretic notation; see for instance Diestel's book~\cite{Diestel05}. All the graphs we consider are undirected and contain neither loops nor multiple edges. If $S \subseteq V(G)$, we define $G[S]=(S, \{\{u,v\} \in E(G)\mid u,v \in S\})$ and $G-S = G[V(G) \setminus S]$. A graph $G=(V,E)$ is a {\it{split graph}} if there is a partition of $V$ into an independent set $I$ and a clique $K$. Split graphs can be generalized as follows. Let $r,\ell$ be two positive integers. A graph $G=(V,E)$ is an {\it{$(r,\ell)$-graph}} if $V$ can be partitioned into at most $r$ independent sets $I^1, I^2, I^3, \ldots, I^r$  and at most $\ell$ cliques $K^1, K^2, K^3, \ldots, K^{\ell}$.



Given a signed graph $G=(V,E)$, a cycle $C$ in $G$ is called {\it{positive}} if the number of negative edges in $C$ is even. Otherwise $C$ is called {\it{negative}}.
Harary~\cite{harary1953}  proved that the absence of {\it{negative}} cycles characterizes balanced graphs.
	
\begin{theorem}[Harary~\cite{harary1953}]\label{theorem2.1}
A signed graph $G$ is balanced if and only if every cycle in $G$ is positive.	
\end{theorem}

The following lemma by Crowston et al.~\cite{crowston2013maximum} is very useful for our purposes, as it gives a lower bound on the maximum size of a balanced signed subgraph of a graph. Let $G=(V,E)$ be a signed graph and let $U, W\subseteq V$. We denote by $E(U,W)$ the subset of $E$ formed by the edges that have one endpoint in $U$ and the other in $W$. We also let  $\beta(G)$ denote the maximum number of edges in a balanced subgraph of $G$.

	\begin{lemma}[Crowston et al.~\cite{crowston2013maximum}]
 Let $G=(V,E)$ be a connected signed graph and let $V=U\cup W$ such that $U\cap W=\emptyset$, $U\neq \emptyset$, and $W\neq \emptyset$. Then $\beta(G)\geq \beta(G[U])+\beta(G[W])+\frac{1}{2}|E(U,W)|$. In addition, if $G[U]$ has $c_1$ components, $G[W]$ has $c_2$ components, $\beta(G[U])\geq {\sf pt}(G[U])+\frac{k_1}{4}$, and $\beta(G[W])\geq {\sf pt}(G[W])+\frac{k_2}{4}$, then $\beta(G)\geq {\sf pt}(G)+\frac{k_1+k_2-(c_1+c_2-1)}{4}$.  	 \label{beta}
	\end{lemma}
	

	We say that a problem is {\it{fixed-parameter tractable}} ({\sf{FPT}})~\cite{FG06,Nie06,DF13,CyganFKLMPPS15} with respect to parameter $k$
if there exists an algorithm that solves the problem  in time $f(k)\cdot n^{O(1)}$, where $f$ is a computable function of $k$ which is independent
of $n$.
	{\it{Kernelization}} is an important
technique used to shrink the size of a given problem instance by means of
polynomial-time data reduction rules until the size of this instance is bounded by a function of
the parameter $k$. The reduced instance is called a problem {\it{kernel}}. Once a problem
kernel is obtained, we know that the problem is fixed-parameter tractable, since
the running time of any brute force algorithm depends on the parameter $k$ only.
The converse is also true: whenever a parameterized problem is {\sf{FPT}}, then it admits a kernel~\cite{FG06,Nie06,DF13,CyganFKLMPPS15}. The natural question to be asked about kernels is whether a parameterized problem admits a {\em polynomial} kernel or not, that is, a kernel of size $k^{O(1)}$.

	In their article, Crowston et al.~\cite{crowston2013maximum} proved that {\sc{Signed Max Cut ATLB}} is {\sf{FPT}} on general graphs by designing an algorithm  running in  time $2^{3k}\cdot n^{O(1)}$. The algorithm applies some reduction rules to the input $(G,k)$ that either answer that $(G,k)$ is a {\sc{Yes}}-instance,  or produce a set $S$ of at most $3k$ vertices such that $G-S$ is a forest of cliques, or equivalently \emph{clique-forest}, which is a graph such that every 2-connected component (that is, each block) is a clique without positive edges. (Note that this definition differs from the usage of this term in the context of chordal graphs.) We state this property formally as follows.

\begin{proposition}[Crowston et al.~\cite{crowston2013maximum}]
\label{coro4.1}\label{coroker}
Let $(G,k)$ be an instance of {\sc{Signed Max Cut ATLB}}. In polynomial time we can conclude that $(G,k)$ is a {\sc{Yes}}-instance or we can find a set $S$ of at most $3k$ vertices such that $G-S$ is a clique-forest without positive edges.
\end{proposition}

%
	
	There are two kinds of reduction rules applied by Crowston et al.~\cite{crowston2013maximum} to an instance $(G,k)$ in order to obtain the kernel of size $O(k^3)$ for {\sc{Signed Max Cut ATLB}}: one-way reduction rules and two-way reduction rules. In a \emph{two-way reduction rule}, the instance $(G',k')$ produced by the reduction rule is {\sl equivalent} to $(G,k)$ (that is, $(G,k)$ is a {\sc{Yes}}-instance if and only if
$(G',k')$ is a {\sc{Yes}}-instance), so these rules can be safely applied to any instance in order to obtain {\sf{FPT}}-algorithms or kernels, as long as the parameter $k$ does not increase, which is always the case in the rules defined by Crowston et al.~\cite{crowston2013maximum}. A two-way reduction rule is {\it{valid}} if it transforms {\sc{Yes}}-instances into {\sc{Yes}}-instances and {\sc{No}}-instances into {\sc{No}}-instances, that is, if it works as it should.

 On the other hand, in a \emph{one-way reduction rule}, the instance $(G',k')$ produced by the reduction rule does not need to be equivalent to $(G,k)$, but only needs to satisfy that if $(G',k')$ is a {\sc{Yes}}-instance, then $(G,k)$ is a {\sc{Yes}}-instance as well. The usefulness of such rules relies on the fact that if after the application of some two-way or one-way reduction rules we obtain an instance $(G',k')$ with $k' \leq 0$, we can safely conclude that the original instance $(G,k)$ is a {\sc{Yes}}-instance. This fact will be used in the linear kernel provided in Section~\ref{sec:linear-kernel}. A one-way reduction rule is {\it{safe}} if it does not transform a {\sc{No}}-instance into a {\sc{Yes}}-instance.

\section{A quadratic kernel on $(r,\ell)$-graphs}
\label{sec:quadratic-kernel}

In this section we show that {\sc{Signed Max Cut ATLB}} admits a quadratic kernel on $(r,\ell)$-graphs for any fixed integers $r,\ell \geq 0$. For this we strongly use the results and the techniques provided by Crowston et al.~\cite{crowston2013maximum}. The kernelization algorithm applies to the instance $(G,k)$ the algorithm given by Proposition~\ref{coro4.1}. This algorithm uses the two-way reduction rules Rule~8, Rule~9, Rule~10, and Rule~11 defined by Crowston et al.~\cite{crowston2013maximum}. Namely, when these rules cannot be applied anymore, either one can directly conclude that $(G,k)$ is a  \textsc{Yes}-instance or, using that the graph cannot be reduced anymore, one can prove that its size is bounded by a function of $k$ only. It is at this point when we exploit the structure of $(r,\ell)$-graphs in order to improve the cubic kernel on general graphs and obtain a quadratic kernel on $(r,\ell)$-graphs. We would like to stress that we do {\sl not} need to obtain algorithmically any particular partition of $V(G)$ into $r$ independent sets and $\ell$ cliques, as such a partition will only be used for the {\sl analysis}, and not by the kernelization algorithm.

This section is organized as follows. For completeness, in Subsection~\ref{sec:rules} we state all the reduction rules and the results from Crowston et al.~\cite{crowston2013maximum} that we need for our purposes. For the sake of providing more intuition on the ideas behind our kernel and because it is simpler, we first present in Subsection~\ref{ap:quadratic-split} the quadratic kernel on split graphs (that is, for $r=\ell= 1$), and we describe in Subsection~\ref{subsec:quadratic-general} its generalization to arbitrary values of $r$ and $\ell$. The main difference is that, in the case of split graphs, it can be seen that we do not have to worry about path vertices in $G-S$.

\subsection{Reduction rules and known results}
\label{sec:rules}

Before stating the two-way reduction rules used in the kernelization algorithm, we first need to introduce some definitions from~\cite{crowston2013maximum}. Let $S$ be a set of vertices as in Proposition~\ref{coro4.1}. For a block $C$ in $G-S$, let $C_{\text{int}}=\{x\in V(C)  \mid N_{G-S}(x)$ $\subseteq V(C)\}$ be the {\it{interior}} of $C$, and let $C_{\text{ext}}=V(C)\backslash C_{\text{int}}$ be the {\it{exterior}} of $C$.
If a block $C$ has only two vertices and these vertices belong to $C_{\text{ext}}$, then $C$ is called a {\it{path block}}. A vertex that belongs only to path blocks is called a {\it{path vertex}}. A block $C$ in $G-S$ is called a {\it{leaf block}} if $|C_{\text{ext}}|\leq 1$. We denote by $N^+$ (resp. $N^-$) the set of neighbors (of a vertex or of a vertex set) adjacent via a positive (resp. negative) edge.


In the following rules, the main idea is that, if some simple local conditions are satisfied, then certain vertices or edges can be deleted from the graph without changing the answer to the problem. Most of these conditions concern blocks or connected components of $G-S$ and their neighborhoods in $S$.

\vspace{.3cm}

\noindent	{\bf{Rule 8.}}{ \it{Let $C$ be a block in $G-S$. If there exists $X\subseteq C_{\text{int}}$ such that $|X| > \frac{|V(C)| + |N_G(X) \cap S|}{2} \geq 1$, $N_G^+(x) \cap S = N_G^+(X) \cap S$ and $N_G^-(x) \cap S = N_G^-(X) \cap S$ for all $x \in X$, then delete two arbitrary vertices $x_1,x_2 \in X$ and set $k' = k$.}}

\vspace{.3cm}

\noindent	{\bf{Rule 9.}}{ \it{Let $C$ be a block in $G-S$. If $|V(C)|$ is even and there exists $X\subseteq C_{\text{int}}$ such that $|X|=\frac{|V(C)|}{2}$ and $N_G(X)\cap S=\emptyset$, then delete a vertex $x\in X$, and set $k'=k-1$.}}

\vspace{.3cm}
	
\noindent	{\bf{Rule 10.}}{ \it{Let $C$ be a block in $G-S$ with vertex set $\left\{x,y,u\right\}$, such that $N_G(u)=\left\{x,y\right\}$. If the edge $\{x,y\}$ is a bridge in $G-u$, delete $C$, add a new vertex $z$, positive edges $\left\{\{z,v\} \mid v\in N^+_{G-u}(\left\{x,y\right\})\right\}$, negative edges $\{    \{z,v\} \mid    v \in N^{-}_{G-u}(\{x,y\})   \}$, and set $k'=k$. Otherwise, delete $u$ and the edge $\{x,y\}$ and set $k'=k-1$.}}

\vspace{.3cm}

An illustration of the application of Rule~10 can be found in Fig.~\ref{fig:Rule10}.

\begin{figure}[h!]
\begin{center} 
\includegraphics[width=.45\textwidth]{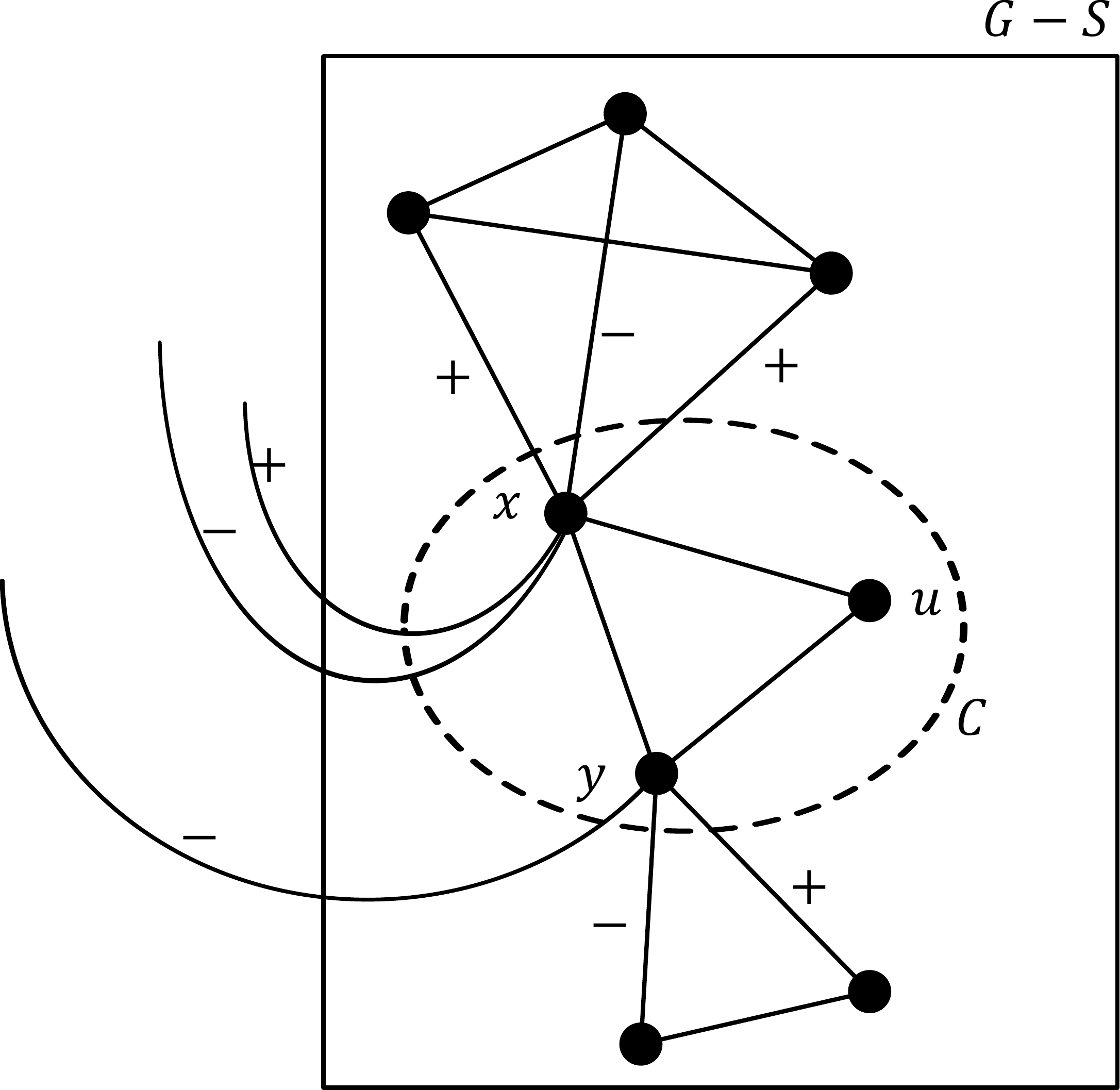}\\
\vspace{.4cm}
  \includegraphics[width=.45\textwidth]{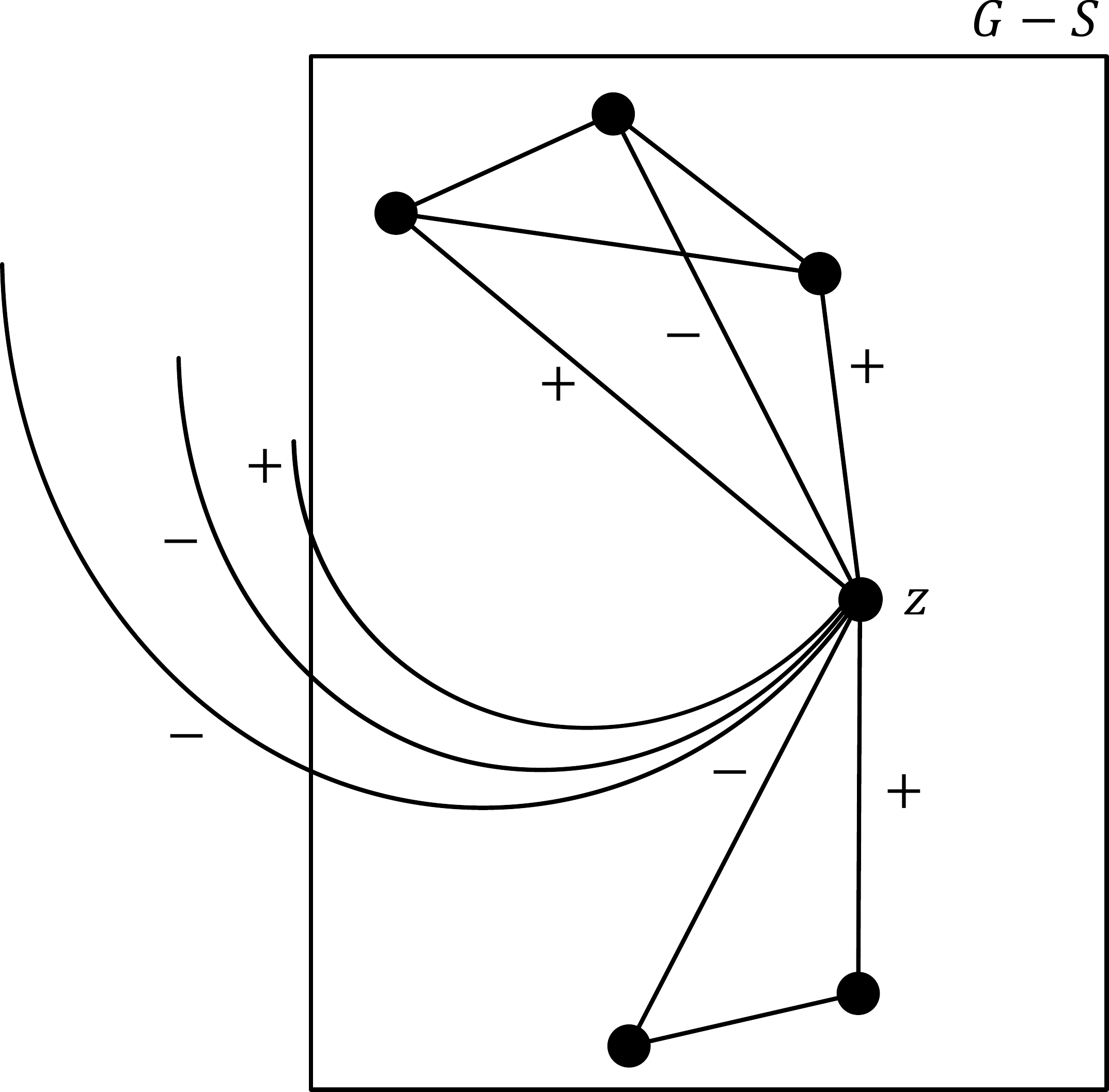}\hspace{1cm}
\includegraphics[width=.45\textwidth]{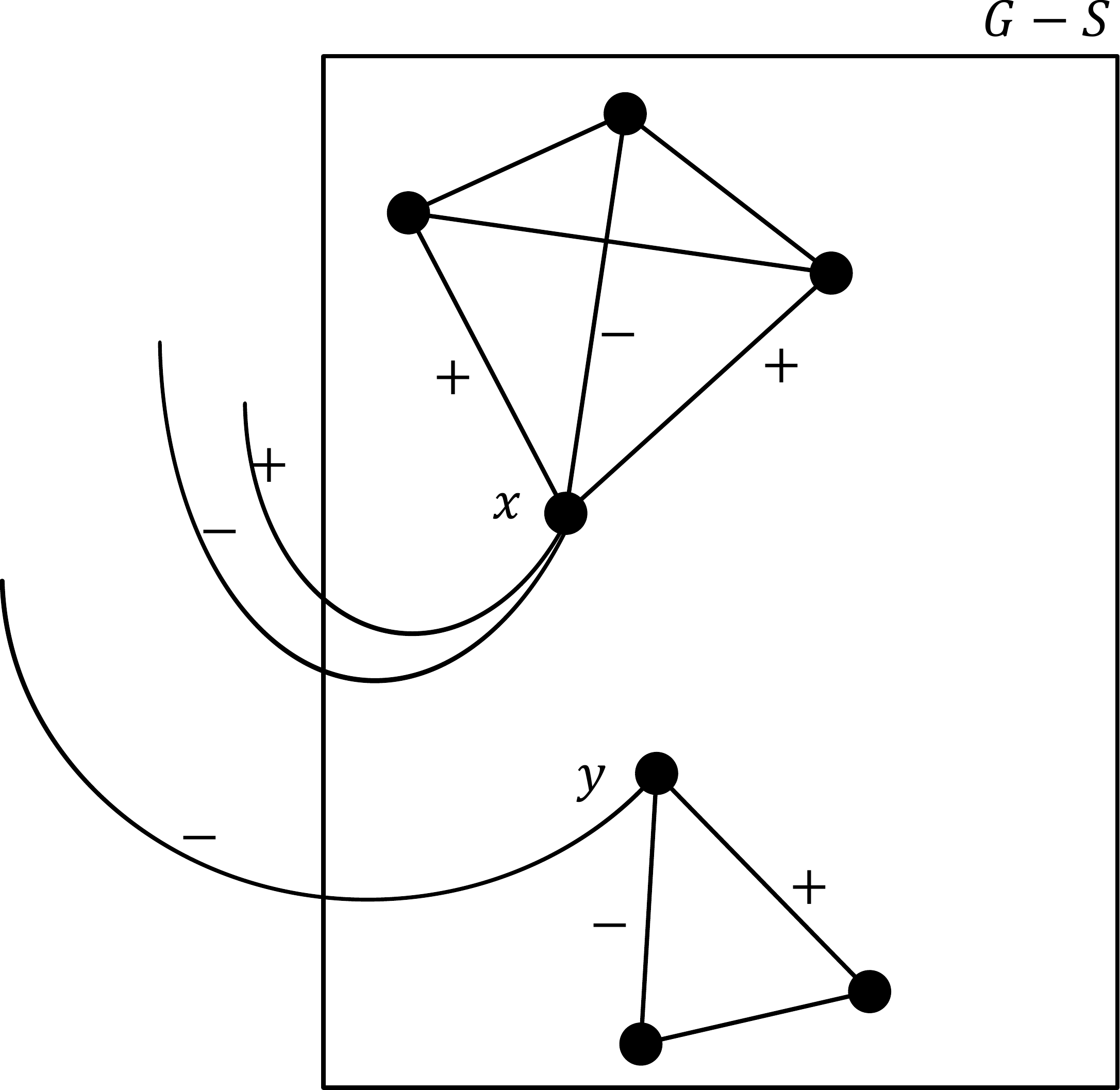}
\caption{Illustration of the application of Rule 10: a block $C$ satisfying the conditions (up), and the two possible resulting graphs after applying the rule (down).} \label{fig:Rule10} 
\end{center}
\end{figure}

\vspace{.3cm}
The \textsc{Max Cut with Weighted Vertices} problem is defined as follows. We are given a graph $G$ with weight functions $w_1:V(G) \to \mathds{N}_0$ and $w_2:V(G) \to \mathds{N}_0$, and an integer $t \in \mathds{N}$, and the question is whether there exists an assignment $f: V(G) \to \{1,2\}$ such that $\sum_{\{x,y\} \in E} |f(x) - f(y)| + \sum_{f(x) =1}w_1(x) + \sum_{f(x) =2}w_12(x) \geq t$.
\vspace{.3cm}

\noindent	{\bf{Rule 11.}}{ \it{Let $T$ be a connected component of $G-S$ only adjacent to a vertex $s \in S$. Form a \textsc{Max Cut with Weighted Vertices} instance on $T$ by defining $w_1(x) = 1$ if $x \in N_G^+(s) \cap T$ ($w_1(x)= 0$ otherwise) and $w_2(y) = 1$ if $y \in N_G^-(s) \cap T$ ($w_2(y)= 0$ otherwise). Let $\beta(G[V(T) \cup \{s\}]) = {\sf pt}(G[V(T) \cup \{s\}]) + \frac{p}{4}$. Then delete $T$ and set $k' = k-p$.}}

\vspace{.3cm}


By~\cite[Lemma 9]{crowston2012max}, the value of $p$ in Rule~11 can be found in polynomial time.

\vspace{.3cm}

Crowston et al.~\cite{crowston2013maximum} proved that the two-way reduction Rules~8-11 are all valid. As mentioned in~\cite{crowston2013maximum}, since it can be easily verified  that none of these  reduction rules increases the number of positive edges, our proof also implies  a kernel of size $O(k^2)$ for {\sc{Signed Max Cut ATLB}} on $(r,\ell)$-graphs.


We now state the results from Crowston et al.~\cite{crowston2013maximum} that we will use in the proofs of Theorem~\ref{thm:quadratic-split} and Theorem~\ref{thm:quadratic-general}.

\begin{lemma}[Lemma 9 in Crowston et al.~\cite{crowston2013maximum}]\label{Crow-Lemma-9}
Let $T$ be a connected component of $G -S$. Then for every leaf block $C$ of $T$, $N_G(C_{\text{int}}) \cap S \neq \emptyset$. Furthermore, if $|N_G(S) \cap V(T)|=1$, then $T$ consists of a single vertex.
\end{lemma}

\begin{lemma}[Corollary 5 in Crowston et al.~\cite{crowston2013maximum}]\label{Crow-Coro-5}
If $\sum_{C \in \mathcal{B}} |N_G(C_{\text{int}}) \cap S| \geq |S| (2|S| - 3 + 2k) + 1$, the instance is a \textsc{Yes}-instance. Otherwise, $\sum_{C \in \mathcal{B}} |N_G(C_{\text{int}}) \cap S| \leq 3k(8k-3)$.
\end{lemma}

\begin{lemma}[Corollary 6 in Crowston et al.~\cite{crowston2013maximum}]\label{Crow-Coro-6}
 $G-S$ contains at most $6k(8k-3)$ non-path blocks and $24k^2 (8k-3)$
path vertices.
\end{lemma}

\begin{lemma}[Corollary 7 in Crowston et al.~\cite{crowston2013maximum}]\label{Crow-Coro-7}
$G - S$ contains at most $12k(8k - 3)$ vertices in the exteriors of
non-path blocks.
\end{lemma}

\begin{lemma}[Lemma 14 in Crowston et al.~\cite{crowston2013maximum}]\label{Crow-Lemma-14}
For a block $C$, if $|V (C)| \geq 2|C_{\text{ext}} |+|N_G (C_{\text{int}}) \cap S|(2|S|+2k +1)$,
then $(G,k)$ is a \textsc{Yes}-instance. Otherwise, $|V (C)| \leq 2|C_{\text{ext}}|+|N_G (C_{\text{int}} ) \cap S|(8k+ 1)$.
\end{lemma}

\subsection{A quadratic kernel on split graphs}
\label{ap:quadratic-split}

In this subsection we present a quadratic kernel for {\sc{Signed Max Cut ATLB}} on split graphs, which contains the main ideas of the kernel described in Theorem~\ref{thm:quadratic-general} for arbitrary $(r,\ell)$-graphs. The main simplification is that in the case of split graphs, we do not have to worry about path vertices in $G-S$.

The proof is based on exploiting the structure given by the set $S$ of Proposition~\ref{coro4.1}. We will strongly use the fact that, as  the input graph is a split graph and this property is hereditary, both $G[S]$ and $G-S$ are split graphs as well.

     \begin{theorem}\label{thm:quadratic-split}
	{\sc{Signed Max Cut ATLB}} on split graphs admits a kernel with $O(k^2)$ vertices.
	\end{theorem}
     \begin{proof}
     Let $(G,k)$ be the given instance, where $G=(V,E)$ is a split graph. Applying the polynomial-time algorithm of Proposition~\ref{coro4.1}, we can either conclude that $(G,k)$ is a \textsc{Yes}-instance, or we obtain a set $S \subseteq V(G)$ such that $G-S$ is a clique-forest without positive edges. Let $S=K^1\cup I^1$ and $G-S=K^2\cup I^2$ be such that $K^1 \cup K^2$ is a clique and $I^1 \cup I^2$ is an independent set. We partition $I^2$ into three subsets: $I^2_0$ containing only isolated vertices in $G-S$, $I^2_1$ with all vertices of degree one in $G-S$, and $I^2_u$ containing all vertices $u \notin I^2_0 \cup I^2_1$ such that $N_{K^2}(u)=K^2$; see Fig.~\ref{fig:quadratic-split} for an illustration. Note that $I^2_u$ can contain at most one vertex. Indeed, if $|K^2| \leq 1$, then the set $I^2_u$ is empty by definition. Otherwise, if $|K^2| \geq 2$, if $I^2_u$ contains two distinct vertices $u_1$ and $u_2$, then $G-S$ is not a clique-forest, contradicting Proposition~\ref{coro4.1}.

\begin{figure}[h]
\begin{center} 
 \includegraphics[width=.45\textwidth]{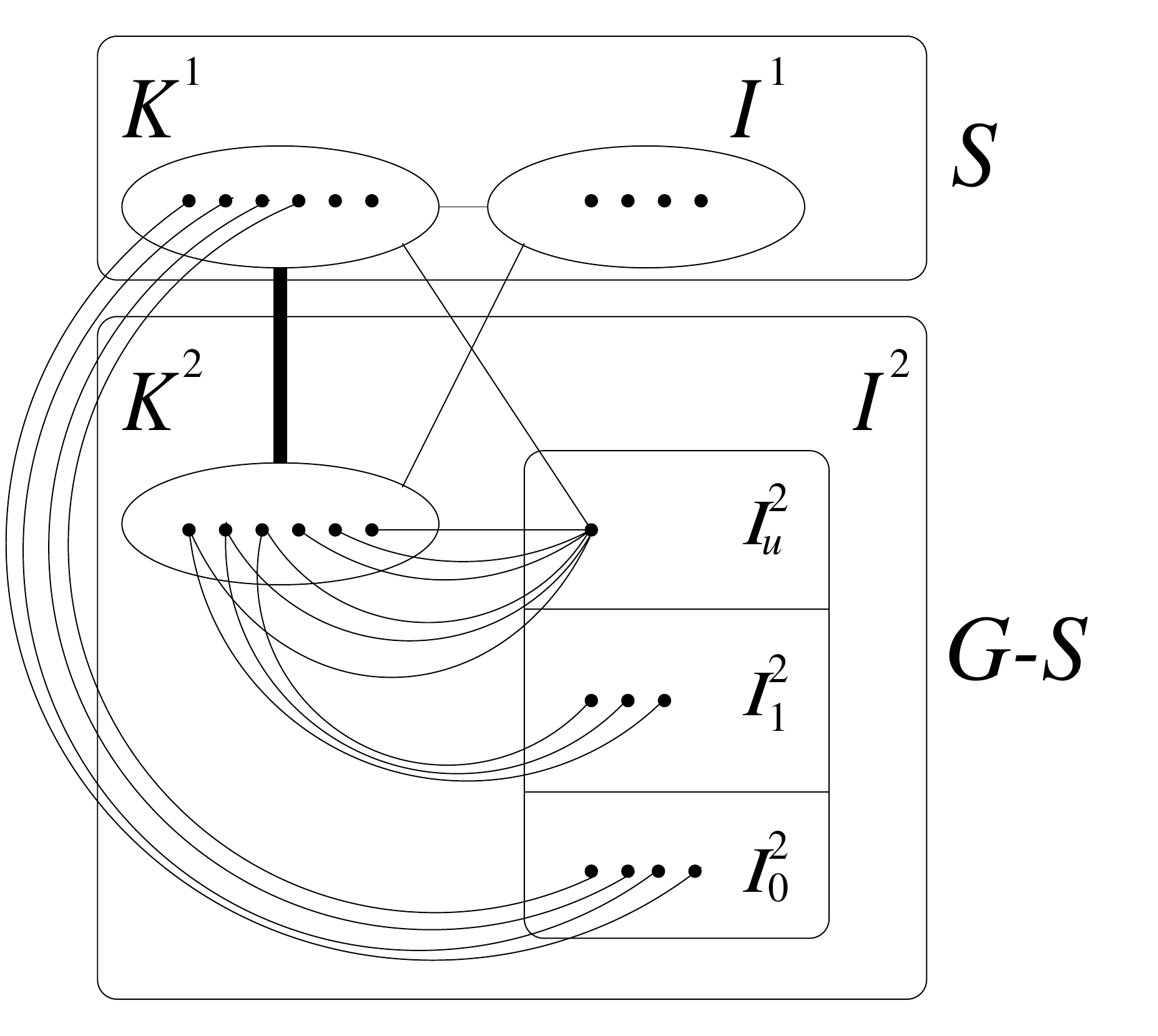}
\caption{Illustration of the proof of Theorem~\ref{thm:quadratic-split}.} \label{fig:quadratic-split} 
\end{center}
\end{figure}

     These sets indeed define a partition of $I^2$, as  if there exists a vertex $v$ in $I^2$ such that $1<d_{K^2}(v)<|V(K^2)|$, we could find a block in $G-S$ that is not a clique. On the other hand, due to the structure of the clique-forest, two maximal cliques can intersect only in one vertex.

     Observe that $G-S$ has at most one path block. Indeed, it is easily seen that the only possible case in which a path block exists in $G-S$ is when $K^2=K_2$ (the complete graph with two vertices), $I^2_u=\emptyset$, and each vertex of $K_2$ has at least one adjacent vertex in $I^2_1$. Hence, $G-S$ has no path vertices since in $G-S$ there is only one possible path block (namely, $K_2$) adjacent to at least two leaf blocks; see Fig.~\ref{fig:path-block} for an illustration.

     \begin{figure}[h!]
\begin{center} 
 \includegraphics[width=.35\textwidth]{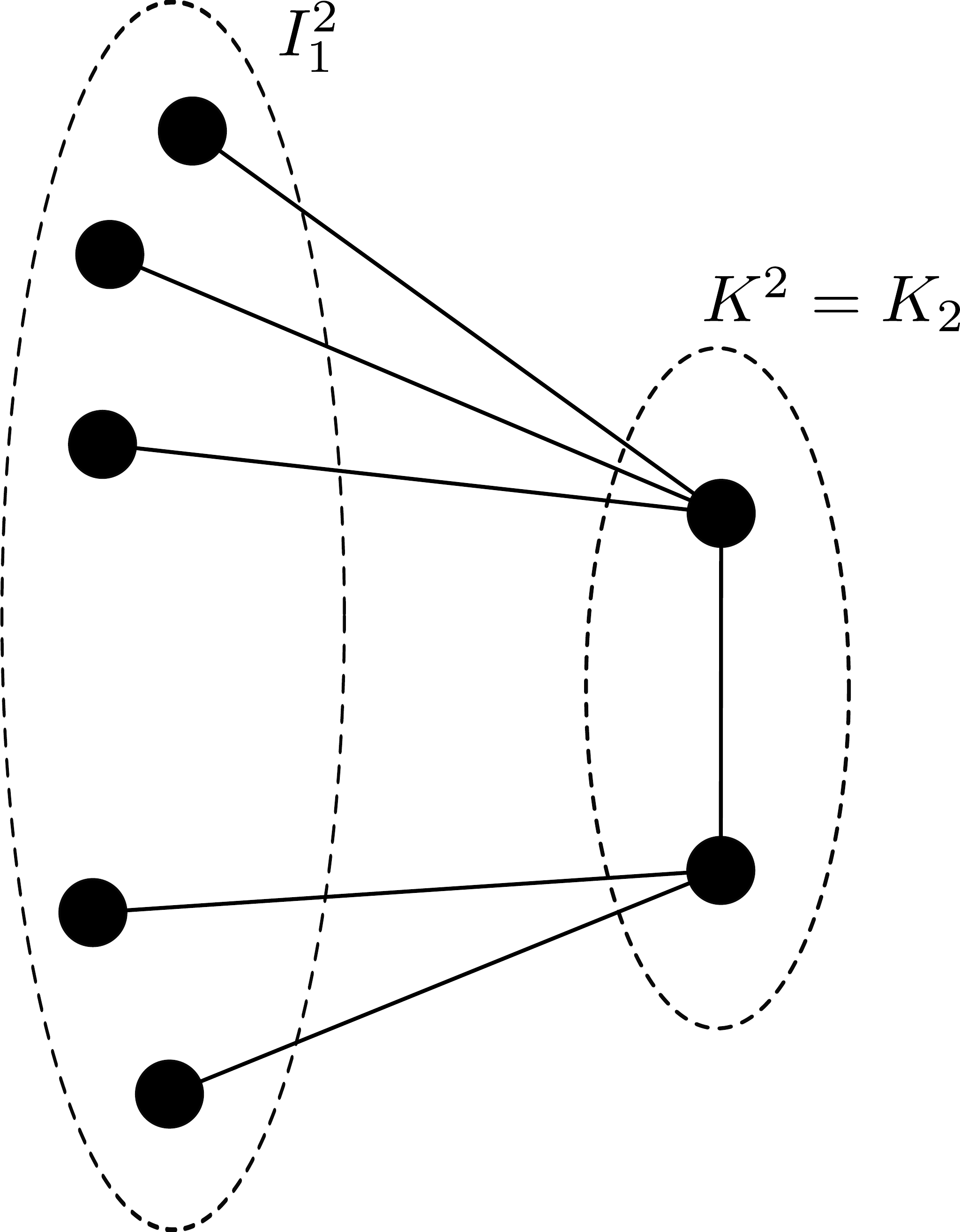}
\caption{Illustration of the unique possible path block in $G-S$.} \label{fig:path-block} 
\end{center}
\end{figure}

\newpage 

     For each vertex $x$ in the set $I^2_1$ there is a leaf block $C^x$ with $|V(C^x)|=2$ and $C^x_{\text{int}}=\left\{x\right\}$. So, Rule 9 could be successively applied to each vertex in $I^2_1$ with no neighbor in $S$, deleting all these vertices and reducing the parameter $k$ accordingly.	We could also delete all the vertices of $I^2_0$ with no neighbor in $S$, because they are isolated vertices and make no contribution to the cut. So in this case we set $k'=k$.

     We note that Rule 10 could be applied only in two cases. In the first case $K^1=\emptyset$ and $K^2$ is a triangle, and in the second case, $K^1=\emptyset$, $K^2=\left\{x,y\right\}$, and $I^2_u=\left\{u\right\}$.

	Using Lemma~\ref{Crow-Lemma-9}, it can be easily seen that after applying Rule 9, the remaining vertices of $I^2_1$ are adjacent to some vertex in $S$. Indeed, otherwise  Rule~9 could be applied again, deleting the vertices $x\in C^x_{\text{int}}$ of $I^2_1$ with $N_G(C^x_{\text{int}})\cap S=\emptyset$, contradicting the hypothesis that the graph $G$ is reduced under all two-way reduction rules; see Fig.~\ref{fig:blocosfolha} for an illustration.

     \begin{figure}[h!]
\begin{center} 
 \includegraphics[width=.38\textwidth]{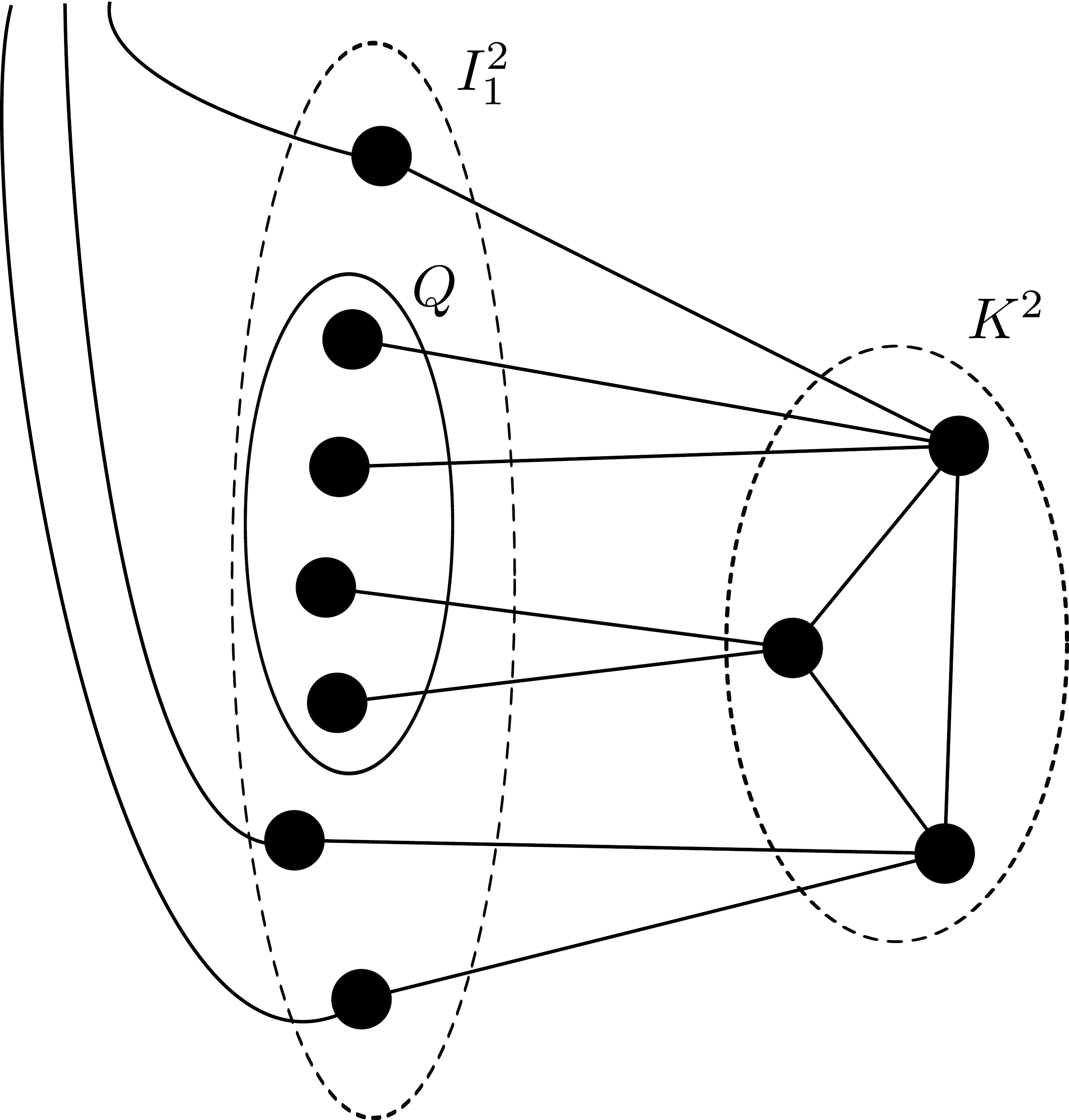}
\caption{The set $Q \subseteq I_1^2$, if it were non-empty, would be deleted by Rule 9.} \label{fig:blocosfolha} 
\end{center}
\end{figure}

	Let $\mathcal{B}$ be the set of non-path blocks of $G-S$. As discussed before,  in a split graph $G$  all blocks of $G-S$, but possibly one, are non-path blocks. Since if there is some path block, we have that $|K^2|= |K_2 |=2$, we may assume in what follows that all blocks of $G-S$ are in $\mathcal{B}$.


By Lemma~\ref{Crow-Coro-5}, either we can identify that $G$ is a {\sc{Yes}}-instance or it follows that $\sum_{C\in\mathcal{B}}|N_{G}(C_{\text{int}})\cap S|\leq 3k(8k-3)$. Using Lemma~\ref{Crow-Lemma-14} it follows that for all blocks $C$ in $G-S$, we have
$$
|V(C)|\leq 2|C_{\text{ext}}|+|N_G(C_{\text{int}})\cap S| \cdot (8k+1).
$$
	Then $|V(G)|\leq |S|+\sum_{C\in\mathcal{B}}{|V(C)|}$ since we may assume that are no path blocks in $G-S$. As $G$ is a split graph, $G-S$ is also a split graph and our previous analysis shows that $\mathcal{B}$ is composed by one big block $C_0$ defined by $K^2$, leaf blocks that are edges which have at least one neighbor in $S$, and possibly isolated vertices. Let $\mathcal{B_*}$ be the set of leaf blocks in $\mathcal{B}$. We can rewrite the previous inequality as
$$
|V(G)|\leq |S|+|V(C_0)|+\sum_{C\in\mathcal{B_*}}{|V(C)|}.
$$

\newpage
In particular, for the big block $C_0$ we have $|V(C_0)|\leq 2|{C_0}_{\text{ext}}|+|N_G({C_0}_{\text{int}})\cap S|(8k+1)\leq 2|{C_0}_{\text{ext}}|+|S|(8k+1)$. And for all leaf blocks  $C \in \mathcal{B_*}$ we have $|V(C)|\leq 2$.
	Then
$$
\sum_{C\in\mathcal{B_*}}{|V(C)|}\leq 2|\mathcal{B_*}|\leq 2\sum_{C\in\mathcal{B_*}}{|N_G(C_{\text{int}})\cap S|}\leq 2\cdot 3k(8k-3)
$$
as a consequence of Lemma~\ref{Crow-Coro-5}, since every block $C \in \mathcal{B_*}$ has a vertex $v \in C_{\text{int}}$ such that $N_S(v)\neq\emptyset$.
	Finally, the number of vertices in $G$ can be bounded as follows:
\begin{eqnarray*}
|V(G)| & \leq & |S|+|S|(8k+1)+6k(8k-3)+2|{C_0}_{\text{ext}}| \\
& \leq & 3k+3k(8k+1)+6k(8k-3)+2|\bigcup_{C\in\mathcal{B}}C_{\text{ext}}|\\
& \leq &  3k+3k(8k+1)+6k(8k-3)+12k(8k-3) \\
& = &168k^2 - 48k = O(k^2),
\end{eqnarray*}
where the last inequality is due to Lemma~\ref{Crow-Coro-7}.
\end{proof}

\subsection{Generalization to arbitrary $(r,\ell)$-graphs}
\label{subsec:quadratic-general}

In this subsection we show how the ideas of the previous subsection can be generalized to arbitrary $(r,\ell)$-graphs. We follow the same strategy as in the proof of Theorem~\ref{thm:quadratic-split}, but we need a number of extra arguments, mostly to deal with the path vertices in $G-S$.

Namely, we will again exploit the structure given by the set $S$ of Proposition~\ref{coro4.1}, and the fact that, as  the input graph is an $(r,\ell)$-graph and this property is hereditary, both $G[S]$ and $G-S$ are $(r,\ell)$-graphs as well. In order to bound the size of the set of path vertices in $G-S$, which we call $\mathcal{P}$ in the proof, we will need three technical claims.

\begin{theorem}\label{thm:quadratic-general}
For every two integers $r,\ell \geq 0$, \textsc{Signed Max Cut ATLB} on $(r,\ell)$-graphs admits a kernel with $O((r+\ell)k^2)$ vertices.
\end{theorem}
\begin{proof}
Let $G$ be an $(r,\ell)$-graph for two integers $r,\ell \geq 0$. Denote by $K^1,K^2,\ldots,K^\ell$ the cliques and $I^1,I^2,\ldots,I^r$ the independent sets of $G$.  We stress again that we just use the partition of $G$ for the analysis, but we do not need to know exactly this partition.  By Proposition~\ref{coro4.1}, we can find a set $S$ with at most $3k$ vertices such that the subgraph $G-S$ is a forest of cliques without positive edges. For an $(r,\ell)$-graph, we have $S=K^1_1\cup\cdots\cup K^\ell_1\cup I^1_1\cup\cdots\cup I^r_1$ and $G-S=K^1_2\cup\cdots\cup K^\ell_2\cup I^1_2\cup\cdots\cup I^r_2$ with $K^i = K^i_1 \cup K^i_2$ and $I^j = I^j_1 \cup I^j_2$ for $1 \leq i \leq \ell$ and $1 \leq j \leq r$. In contrast to split graphs (see  Subsection~\ref{ap:quadratic-split}), we may have path vertices in $G-S$.

In order to bound the size of $V(G)$, we bound separately the total size of non-path blocks and the number of path vertices (that is, the only vertices that do not belong to non-path blocks) in $G-S$. Denote again by $\mathcal{B}$ the set of non-path blocks, by $\mathcal{B}_b$ the set of non-path blocks containing at least two vertices  from the set $K^1 \cup K^2 \cup \cdots \cup K^\ell$, and $\mathcal{B}_s=\mathcal{B}\backslash\mathcal{B}_b$.
	The structure of clique-forests implies that $|\mathcal{B}_b|\leq \ell$, and by Lemma~\ref{Crow-Lemma-14}, for every block $C\in\mathcal{B}_b$ either we can conclude that $(G,k)$ is a \textsc{Yes}-instance or we have that 
 $$
 |V(C)|\leq 2|C_{\text{ext}}|+|N_{G}(C_{\text{int}})\cap S| \cdot (8k+1)\leq 2|C_{\text{ext}}|+|S|(8k+1).
 $$

Then we can write
\begin{eqnarray*}
\sum_{C\in\mathcal{B}_b}{|V(C)|} & \leq &  2 \sum_{C\in\mathcal{B}_b}{|C_{\text{ext}}|}+\sum_{C\in\mathcal{B}_b}{|S|(8k+1)}\\
 & \leq & 2 \ell  \cdot |\bigcup_{C\in\mathcal{B}_{\ell}}C_{\text{ext}}|+\ell|S|(8k+1) \\
 & \leq & 24 \ell k(8k-3)+3\ell k(8k+1),
\end{eqnarray*}
where in the second inequality we have used twice that $|\mathcal{B}_b|\leq \ell$, and the last inequality follows from Lemma~\ref{Crow-Coro-7}.

For each block $C\in\mathcal{B}_s$ we have $|V(C)|\leq r + 1$, as $C$ contains only vertices of the $r$ independent sets $I^1,I^2,\ldots,I^r$ plus possibly one vertex of some of the cliques. Indeed, otherwise if  $|V(C)|> r + 1$ then there are two adjacent vertices in the same independent set, which is impossible. Therefore, by~Lemma~\ref{Crow-Coro-6}, we can bound  the sum of the number of vertices in $C\in\mathcal{B}_s$ as $\sum_{C\in\mathcal{B}_s}{|V(C)|}\leq (r+1)\cdot 6k(8k-3)$.

Therefore, by combining the bounds from the two above paragraphs we conclude that
\begin{equation}\label{eq:non-path-blocks}
\sum_{C\in\mathcal{B}}{|V(C)|} = \sum_{C\in\mathcal{B}_b}{|V(C)|}+\sum_{C\in\mathcal{B}_s}{|V(C)|} \leq (2r + 9\ell)3k(8k+1) =: g(k,r,\ell).
\end{equation}

It just remains to bound the number of path vertices in the blocks of $G-S$. Note that we cannot use directly the bound on the number of path vertices in $G-S$ given by~Lemma~\ref{Crow-Coro-6}, as this bound is $O(k^3)$, which is too much for our purposes.

Let $\mathcal{P}$ be the set of path vertices of $G-S$. In order to bound the size of $\mathcal{P}$, we argue separately about the components in $G[\mathcal{P}]$ having at most two vertices (Claim~\ref{lemma4.8}) and the remaining components (Claim~\ref{claim:2-3}), where $G[\mathcal{P}]$ denotes the subgraph of $G$ induced by $\mathcal{P}$.

We call an {\it{isolated vertex}} (resp. \emph{isolated edge}) a connected component of $G[\mathcal{P}]$ of size 1 (resp. 2).
We show in Claim~\ref{lemma4.8}  that the number of isolated vertices and isolated edges in $G[\mathcal{P}]$ are both bounded by the number of non-path blocks, which is at most $6k(8k-3)$ by~Lemma~\ref{Crow-Coro-6}.

	\begin{claimN}\label{lemma4.8}
The number of isolated vertices in $G[\mathcal{P}]$ and the number of isolated edges in $G[\mathcal{P}]$ are upper-bounded by the number of non-path blocks.	
	\end{claimN}
	\begin{proof}
	Let us first deal with the number of isolated vertices in $G[\mathcal{P}]$, which we denote by $p_0$. We show by induction on $p_0$ that this number is at most the number of non-path blocks minus one. Recall that a vertex is isolated if and only if it is exclusively adjacent to non-path vertices. We contract each non-path block into a  black vertex and represent each path vertex by a white vertex. So the forest of cliques in $G-S$ is now represented by a simple forest colored by two colors: black and white.
If $p_0=1$, then there are at least two black vertices adjacent to the single white vertex of $G[\mathcal{P}]$. So the property is valid and the basis of the induction is proved. Suppose that the property is valid for $p_0=i > 1$, and consider $G-S$ such that the number of isolated vertices is $i+1$. Choose arbitrarily one of these isolated vertices. Without loss of generality, we can suppose that $G-S$ is connected. Otherwise we can deal with the connected components of $G-S$ separately to achieve the same result. The removal of this white vertex disconnects the tree, increasing the number of connected components. For all connected components the property holds by induction hypothesis, that is, the number of isolated vertices in each connected component is at most the number of non-path blocks in this component minus one. So, the number of isolated vertices in the original graph is at most the total number of non-path blocks minus one, and the first part of the claim follows.

As for the number of isolated edges in $G[\mathcal{P}]$, similarly to above,  we contract each non-path block into a black vertex and represent each path vertex as a white vertex. An isolated edge is represented by an edge that has white endpoints. We identify all these edges and contract each of them into a red vertex. As these edges are isolated,  the red vertices have only black neighbors. Then, proceeding by induction as we did for the isolated vertices, we can show that the number of isolated edges is upper-bounded by the number of non-path blocks.
	\end{proof}

Let $\mathcal{P}_1$ be the subset of $\mathcal{P}$ formed by the non-isolated vertices. Note that the minimum degree of the vertices in $G[\mathcal{P}_1]$, which we denote by $\delta (G[\mathcal{P}_1])$, is such that $\delta (G[\mathcal{P}_1])\geq 1$, so clearly the number of edges of $G[\mathcal{P}_1]$ is at least $\frac{|\mathcal{P}_1|}{2}$; this value is attained when $G[\mathcal{P}_1]$ consists of disjoint edges. We need to improve this bound of $\frac{|\mathcal{P}_1|}{2}$ for being able to state Equation~(\ref{eq:P1}) below, and we can easily do it by forbidding this latter case.


	\begin{claimN}\label{claim:2-3}
If each connected component of $G[\mathcal{P}_1]$ has at least three vertices, then $|E(G[\mathcal{P}_1])|\geq \frac{2|\mathcal{P}_1|}{3}$.
	\end{claimN}
	\begin{proof} Without loss of generality, we can suppose that $G[\mathcal{P}_1]$ is connected by the same reason mentioned in the proof of Claim~\ref{lemma4.8}. Let $p=|\mathcal{P}_1|$. As $G[\mathcal{P}_1]$ is a tree, we have $|E(G[\mathcal{P}_1])|= p-1$. If $p\geq 3$ then $\frac{p}{3}\geq 1$, that is, $p-\frac{2p}{3}\geq 1$ and $p-1\geq\frac{2p}{3}$ as desired. \end{proof}

The induced subgraph $G[\mathcal{P}_1]$ is a forest, so by Theorem~\ref{theorem2.1} $G[\mathcal{P}_1]$ is balanced. Let $e_p$ denote the number of edges in $G[\mathcal{P}_1]$. Then $\beta(G[\mathcal{P}_1])=e_p$  and ${\sf pt}(G[\mathcal{P}_1])=\frac{e_p}{2}+\frac{p-1}{4}$ if we suppose $G[\mathcal{P}_1]$ to be connected, where $p = |\mathcal{P}_1|$. Claim~\ref{claim:2-3} implies that
\begin{equation}\label{eq:P1}
\beta(G[\mathcal{P}_1])-{\sf pt}(G[\mathcal{P}_1])=e_p-\frac{e_p}{2}-\frac{p-1}{4}=
\frac{2e_p-p+1}{4}\geq \frac{\frac{4p}{3}-p+1}{4} \geq \frac{\frac{p}{3}}{4}.
\end{equation}

	\begin{claimN}\label{claim:number-comp}
Let $(G,k)$ be an instance of {\sc{Signed Max Cut ATLB}}. If $\mathcal{P}$ is the set of path vertices in $G-S$, then the number of connected components of $G[\mathcal{P}]$ is at most $6k(8k-3)$. If ${\cal W}=G-\mathcal{P}$, then the number of connected components of ${\cal W}$ is  at most $6k(8k-2)$.
	\end{claimN}
\begin{proof}
 The same ideas used in the proof of Claim~\ref{lemma4.8} easily imply that the number of connected components of $G[\mathcal{P}]$ is at most $6k(8k-3)$. Since ${\cal W}\cap(G-S)$ is the complement of $G[\mathcal{P}]$ in $G-S$, each vertex in ${\cal W} \cap(G-S)$ belongs to a non-path block. And as the size of $S$ is at most $3k$, by~Lemma~\ref{Crow-Coro-6} the number of connected components of ${\cal W}$ is  at most $3k + 6k(8k-3) \leq  6k(8k-2)$. \end{proof}

We are now ready to piece everything together and conclude the proof of the theorem.

\vspace{.3cm}

	
 Due to Claim~\ref{lemma4.8}, by incurring an additive term of $18k(8k-3)$ to the size of the kernel we may assume that each component of $G[\mathcal{P}]$ has at least three vertices, so the hypothesis of Claim~\ref{claim:2-3} holds. We apply Lemma~\ref{beta} to the graph $G$ with $U=\mathcal{P}$. We have $c_1$ connected components of $G[\mathcal{P}]$, $c_2$ connected components of $W$, and $\beta(G)\geq {\sf pt}(G)+\frac{k_1+k_2-(c_1+c_2-1)}{4}$, where $\beta(G[\mathcal{P}])\geq {\sf pt}(G[\mathcal{P}])+\frac{k_1}{4}$ and $\beta(W)\geq {\sf pt}(W)+\frac{k_2}{4}$. It holds of course that $k_2 \geq 0$. By Equation~(\ref{eq:P1}) we have $k_1\geq \frac{p}{3}$ for $p=|\mathcal{P}_1|$ and, by Claim~\ref{claim:number-comp}, it holds that $c_1 + c_2 \leq 12k(8k-2) \leq 96k^2$. Then $\beta(G)\geq {\sf pt}(G)+\frac{\frac{p}{3}-192k^2+1}{4}$. Therefore, if $\frac{p}{3}-192k^2+1\geq k$, then  $(G,k)$ is a {\sc{Yes}}-instance. Otherwise, $\frac{p}{3}\leq k-1+192k^2$, that is, we may assume that $p\leq 576k^2+3k-3$. 


Therefore, combining Equation~(\ref{eq:non-path-blocks}) with the above discussion we can conclude that
$$
|V(G)| \leq   |S| + \sum_{C\in\mathcal{B}}{|V(C)|} + |\mathcal{P}| \leq 3k + g(k,r,\ell) + 18k(8k-3) + p = O(k^2) .
$$
Finally, note that in the whole proof, when we use some result of Crowston et al.~\cite{crowston2013maximum} in which one of the possible outputs is that $(G,k)$ is a \textsc{Yes}-instance (like~Lemma 14 or to check whether $p$ is at most $576k^2+3k-3$ or not), the condition to be checked concerns just the blocks of $G-S$, which we can obtain in polynomial time by Proposition~\ref{coro4.1}. Thus, as we claimed, we do {\sl not} need to compute any partition of $V(G)$ into cliques $K^1,K^2,\ldots,K^\ell$  and independent sets $I^1,I^2,\ldots,I^r$. This concludes the proof of the theorem.\end{proof}

\section{A linear kernel on subclasses of split graphs}
\label{sec:linear-kernel}

%

With the objective of improving the quadratic kernel on $(r,\ell)$-graphs presented in Section~\ref{sec:quadratic-kernel}, in Subsection~\ref{subsec:linear-kernel} we provide a linear kernel on a smaller graph class, namely on the subclass of $d^*$-split graphs for every integer $d \geq 1$, which is defined as follows. A graph $G=(V,E)$ is a \emph{$d^*$-split graph} if $V$ can be partitioned into a clique $K$ and an independent set $I$ such that every vertex in $K$ has at least one neighbor in $I$, and every vertex in $I$ has degree at most $d$.


We first prove in Subsection~\ref{ap:NP-hard} that the problem remains {\sf{NP}}-hard restricted to this class of graphs for every $d \geq 2$. In fact, we prove that even {\sc{Max Cut}} is {\sf{NP}}-hard, and this result is tight in terms of $d$. Indeed, for $d = 0$, a $0^*$-split graph is an independent set and therefore {\sc{Max Cut}} is trivial. For $d = 1$, the claim is given in the following simple observation.

\begin{remark}{\sc{Max Cut}} on $1^*$-split graphs can be solved in polynomial time.
\end{remark}
\begin{proof} Note that by the definition of  $1^*$-split graphs, such a graph has a very precise structure. Namely, it consists of an arbitrarily large clique $K$ such that every vertex $v$ in $K$ has a private neighbor $v'$ in $I$ (that is, a vertex of degree one). Then, an optimal solution of {\sc{Max Cut}} consists of a balanced partition of the clique, and for each $v \in K$, we place its private neighbor $v'$ in the opposite part of $v$.
\end{proof}

\subsection{Max Cut is NP-hard on $d^*$-split graphs}
\label{ap:NP-hard}

In this subsection we prove that \textsc{Max Cut}, which is a particular case of {\sc{Signed Max Cut}}, is {\sf{NP}}-hard on $d^*$-split graphs for every fixed $d \geq 2$.
We first prove easily that {\sc{Max Cut}} is {\sf{NP}}-hard on
graphs without universal vertices (a \emph{universal} vertex is a vertex adjacent to every other vertex in the graph). Then, using a proof  from Bodlaender and Jansen~\cite{bodlaender} on the complexity of \textsc{Max Cut}, we prove that this problem is still {\sf{NP}}-hard on $2^*$-split graphs and, consequently, {\sf{NP}}-hard on $d^*$-split graphs for every fixed $d\geq 2$.



\begin{lemma}\label{lemma4.3}
{\sc{Max Cut}} is {\sf{NP}}-hard on graphs without universal vertices.		
\end{lemma}
\begin{proof}
We use the well known fact that {\sc{Max Cut}} is {\sf{NP}}-hard on general graphs~\cite{garey1976some} to show that this problem is still {\sf{NP}}-hard on graphs without universal vertices. For this, we consider a graph $G$ as instance of  {\sc{Max Cut}}, and construct a new graph $G'=2G$ made of two disjoint copies of $G$. Obviously $G'$ has no universal vertices because it is disconnected.
Let $\textsc{MaxCut}(G)$ be the size of a maximum cut for $G$ and $\textsc{MaxCut}(G')$ be the size of a maximum cut for $G'$. We claim that $\textsc{MaxCut}(G')=2 \cdot \textsc{MaxCut}(G)$. Indeed, $\textsc{MaxCut}(G')\geq 2 \cdot \textsc{MaxCut}(G)$ because if $[V_1,V_2]$ is a cut in $G$ of maximum size, then $2 \cdot \textsc{MaxCut}(G)$ is the size of the cut $[2V_1,2V_2]$; see Fig.~\ref{fig:no-universal-vertices} for an illustration. On the other hand, if $\textsc{MaxCut}(G')>2 \cdot \textsc{MaxCut}(G)$, then as $G'$ is made of two disjoint copies of $G$, necessarily $G$ contains a cut of size greater than $\textsc{MaxCut}(G)$, a contradiction.
\end{proof}

\begin{figure}[h]
\vspace{-1.2cm}
\begin{center}
 \includegraphics[width=.82\textwidth]{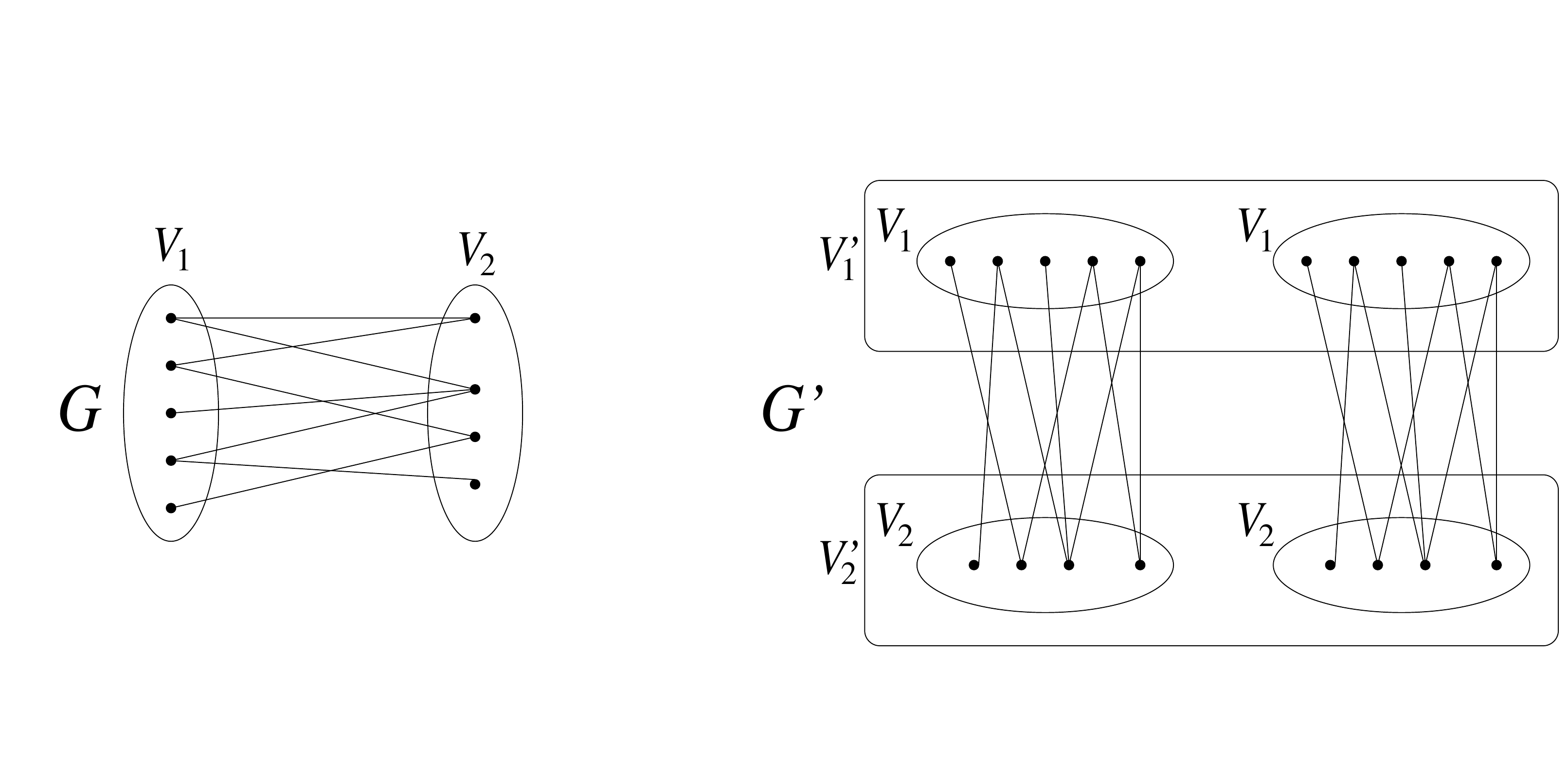}
\vspace{-.35cm}\caption{Reduction to show that \textsc{Max Cut} is {\sf{NP}}-hard on graphs without universal vertices.} \label{fig:no-universal-vertices}\vspace{-.2cm}
\end{center}
\end{figure}

The construction given in the proof of the next theorem is exactly the same as the one given by  Bodlaender and Jansen in~\cite[Theorem 3.1]{bodlaender}, but we reproduce it here because we prove that if we start with a graph without universal vertices (we can do so thanks to Lemma~\ref{lemma4.3}), then the graph constructed in the reduction is a $2^*$-split graph.

		
\begin{theorem}\label{theorem4.4}
{\sc{Max Cut}} is {\sf{NP}}-hard on $2^*$-split graphs.		
\end{theorem}
\begin{proof}
Let a graph $G=(V,E)$ be given and let $\overline{G}=(V,\overline{E})$ be the complement of $G$. Let $H=(V\cup \overline{E}, F)$, where $F=\left\{(v,w)\mid v,w\in V,v\neq w \right\}\cup $ $\left\{(v,e) \mid v\in V, e\in \overline{E}, v\right.$ is an endpoint of edge $\left. e\right\}$. Then $V$ forms a clique, $\overline{E}$ forms an independent set in $H$, and every edge-representing vertex $e$ is connected to the vertices that represent its endpoints. Therefore $H$ is a split graph in which all the vertices in the independent set have  degree exactly two. We claim that $G$ allows a partition with at least $p$ cut edges if and only if $H$ allows a partition with at least $2|\overline{E}|+p$ cut edges.
		
		Suppose first we have a partition $(W_1,W_2)$ of $G$ with at least $p$ cut edges. We partition the vertices of $H$ as follows: partition $V$ as in the partition of $G$; for every $e\in \overline{E}$, if both endpoints of $e$ belong to $W_1$, then put $e$ in $W_2$, otherwise put $e$ in $W_1$. It is easy to see that this partition gives the desired number of cut edges.
		
		Now suppose we have a partition $(W_1,W_2)$ of $H$ with at least $2|\overline{E}|+p$ cut edges. Partition the vertices of $G$ into two subsets: $W_1\cap V$ and $W_2\cap V$. This partition gives the desired number of cut edges. This can be noted as follows: for every edge $\{v,w\}\in E$, we have one cut edge in $H$ if $\left\{v,w\right\}$ is a cut edge in $G$, otherwise we have no cut edge. For every edge $e=\left\{v,w\right\}\in \overline{E}$, we have that out of the three edges $\left\{v,w\right\}, \left\{e,v\right\}, \left\{e,w\right\}$, exactly two will be cut edges. Hence, the total number of cut edges in $H$ equals the number of cut edges in $G$ plus $2|\overline{E}|$. Note that $H$ can clearly be constructed from $G$ in polynomial time.

By Lemma~\ref{lemma4.3}, we can assume that the graph  $G=(V,E)$ has no universal vertices. Then the graph $\overline{G}$ constructed above has no isolated vertices, and therefore every vertex $v\in V$ has at least one neighbor in $\overline{E}$. Since every vertex in $\overline{E}$ has degree two, the graph $H$ constructed in the above reduction is indeed a $2^*$-split graph, and the theorem follows.
\end{proof}
		

Since for every $d \geq 0$, the class of  $d^*$-split graphs contains the class of $2^*$-split graphs, the following corollary is a direct consequence of Theorem~\ref{theorem4.4}.
		
\begin{corollary}\label{cor:NP-hard}
{\sc{Max Cut}} is {\sf{NP}}-hard on $d^*$-split graphs for every fixed $d\geq 2$.				 
\end{corollary}

\subsection{The kernelization algorithm}
\label{subsec:linear-kernel}

As we will see in the proof of Theorem~\ref{thm:linear-d-split} below, our kernelization algorithm is the simplest possible algorithm that one could imagine, as it does {\sl nothing} to the input graph; its interest lies on the analysis of the kernel size,  which uses the following two new reduction rules.

The first one  is a one-way reduction rule that is a generalization of the one-way Rule~6 given by Crowston et al.~\cite{crowston2013maximum}.

\vspace{.3cm}

\noindent {\bf{Rule $\mathbf{6^+}$.}} {\it{Let $G$ be a connected graph. If $v\in V(G)$ and $u_1,\ldots,u_c$ are pairwise non-adjacent neighbors of $v$ such that $c \geq 2$ and $G - \{v,u_1,\ldots,u_c\}$ is connected, then delete $v,u_1,\ldots,u_c$ and set $k'=k-c+1$.}}

\vspace{.3cm}

Note that Rule~6 of Crowston et al.~\cite{crowston2013maximum} corresponds exactly to the case $c=2$ of Rule~$6^+$ above. 


\begin{lemma}\label{lem:Rule6+-safe}\label{lem:6+-safe}
Rule~$6^+$ is safe.		
\end{lemma}
		\begin{proof}
Let $v\in V(G)$ and $u_1,\ldots,u_c$ be as in the description of Rule~$6^+$ and let $P=G[\left\{v,u_1,\ldots,u_c\right\}]$. Note that ${\sf pt}(P)=\frac{c}{2}+\frac{c}{4}=\frac{3c}{4} $ and since $P$ is a tree (namely, a star), $P$ is a balanced graph by Theorem~\ref{theorem2.1} whatever the signs of its edges. Therefore, $\beta(P)=c={\sf pt}(P)+\frac{c}{4}$. Let $G'$ be the graph obtained from $G$ by the deletion of $v,u_1,\ldots,u_c$, so $G'$ is connected by hypothesis.
Suppose that $\beta(G')\geq {\sf pt}(G')+\frac{k'}{4}$, where $k'=k-c+1$. Then, by Lemma~\ref{beta}, $\beta(G)\geq {\sf pt}(G)+ \frac{k'+c-1}{4}={\sf pt}(G)+\frac{k}{4}$, and therefore $G$ is a {\sc{Yes}}-instance. 		
		\end{proof}

The second new rule is a simple two-way reduction rule, which just eliminates vertices of degree 1.
\vspace{.25cm}

\noindent {\bf{Rule A}.} {\it{Let $G$ be a connected graph, and let $v$ be a vertex of degree 1 in $G$. Then delete vertex $v$ and set  $k'=k-1$.}}



\begin{lemma}\label{lem:RuleA}
Rule A is valid.
\end{lemma}
\begin{proof} Let $G$ be a connected graph, let $v$ be a vertex of degree 1 in $G$, and let $G'$ be the graph obtained from $G$ by applying Rule A on vertex $v$. Since the edge containing $v$ in belongs to every optimal balanced subgraph of $G$ (regardless of its sign), it holds that $\beta(G) = \beta(G') +1$. On the other hand, we have that ${\sf pt}(G) = \frac{m}{2}+\frac{n-1}{4}$ and ${\sf pt}(G') = \frac{m-1}{2}+\frac{n-2}{4}$, and therefore $\beta(G) \geq {\sf pt}(G) + \frac{k}{4}$ if and only if $\beta(G') = \beta(G) -1 \geq {\sf pt}(G) + \frac{k}{4}  -1 = {\sf pt}(G') + \frac{k'}{4}$. Thus, $(G,k)$ is a \textsc{Yes}-instance of \textsc{Signed Max Cut ATLB} if and only if $(G',k')$ is.
\end{proof}

We are now ready to prove the main result of this subsection.

\begin{theorem}
\label{thm:linear-d-split} For any fixed integer $d \geq 1$, {\sc{Signed Max Cut ATLB}} on $d^*$-split graphs admits a kernel with at most $4(d+1)k$ vertices.
\end{theorem}
\begin{proof} Given a signed $d^*$-split graph $G=(V,E)$ and an integer $k$, let $(K,I)$ be an arbitrary partition of $V$ into a clique $K$ and an independent set $I$ such that every vertex in $I$ has degree at most $d$ and every vertex in $K$ has at least one neighbor in $I$. We note that we do not even need to compute this partition $(K,I)$ of $V$, as we will just use it for the analysis.

Let $K_h$ be the set of vertices in $K$ that have at least two neighbors in $I$, let $K_s = K \setminus K_h$ (so by hypothesis, every vertex in $K_s$ has exactly one neighbor in $I$), let $I_h = N_{I}(K_h)$, and let $I_s = I \setminus I_h$. Note that since we assume that $G$ is connected, it holds that $|I_s| \leq |K_s| \leq |K|$. We now state two claims that will allow us to certificate in some cases that $(G,k)$ is a \textsc{Yes}-instance.

%
%
%
%



\begin{claimN}\label{claim:clique}
If $|K| \geq (d+1)k$, then $(G,k)$ is a \textsc{Yes}-instance.
\end{claimN}

\begin{proof} We apply to the input $(G,k)$ the following procedure, assuming that $G$ is connected (we stress that this algorithm is only used for the analysis; as mentioned before we do not modify our input graph at all):

\begin{enumerate}
\item If the current graph contains a vertex $v$ in the clique with $c \geq 2$ neighbors in the independent set, let $N$ be this set of neighbors, and do the following:
    \begin{itemize}
    \item[$\circ$] Apply Rule~$6^+$ to $\{v\} \cup N$, thus removing $\{v\} \cup N$ from the current graph and setting $k \leftarrow k - c +1$. Note that the resulting graph is clearly connected, so Rule~$6^+$ can indeed be applied to $\{v\} \cup N$.
     \item[$\circ$] Go back to Step~1.
    \end{itemize}
\item Apply exhaustively Rule A to the current graph, removing all vertices of degree 1. Clearly this operation preserves connectivity, and by Lemma~\ref{lem:RuleA} it produces an equivalent instance.
\item If the current graph contains a vertex $v$ in the clique with exactly 1 neighbor $u$ in the independent set, do the following:
    \begin{itemize}
    \item[$\circ$] Let $w$ be a vertex in the clique that is non-adjacent to $u$; since the degree of $u$ is at most $d$, such a vertex $w$ is guaranteed to exist as long as the size of the current clique is at least $d+1$.
    \item[$\circ$] Note that the removal of the vertices $\{v,u,w\}$ does not disconnect the current graph. Indeed, assume for contradiction that there exists a vertex $z$ in the independent set that gets disconnected after the removal of $\{v,u,w\}$. Since Rule A cannot be applied anymore to the current graph, $z$ cannot have degree 1, hence necessarily $z$ is adjacent to both $w$ and $v$. But then $v$ has at least 2 neighbors $u$ and $z$ in the independent set, and this contradicts the fact that we are in Step~3, since Step~1 could be applied to $v$ together with its neighborhood in the independent set.
     \item[$\circ$] Therefore, we can apply Rule~$6^+$ to  $\{v,u,w\}$, thus removing them from the current graph and setting $k \leftarrow k - 1$.
    \item[$\circ$] Go back to Step~2.
    \end{itemize}
 \item Stop the procedure.
\end{enumerate}


\noindent We now claim that if $|K| \geq (d+1)k$, then the instance $(G',k')$ output by the above procedure satisfies that $k' \leq 0$, and since $(G',k')$ is obtained from $(G,k)$ by applying the one-way reduction Rule~$6^+$ and the two-way reduction Rule A, Lemma~\ref{lem:6+-safe} and Lemma~\ref{lem:RuleA} imply that  $(G',k')$ is also a \textsc{Yes}-instance, as we wanted to prove. Indeed, suppose that Step~1 has been applied $p_1$ times, that Rule~A in Step~2 has been applied $p_2$ times, and that Step~3 has been applied $p_3$ times. Since in Step~1 it holds that $c \geq 2$, it follows that $k' \leq k - (p_1 + p_2 + p_3)$, so it suffices to prove that $p_1 + p_2 + p_3 \geq k$.

With slight abuse of notation, let us denote by $K^*$ the set of vertices in the current clique that have some neighbor in the independent set. Note that as long as $|K^*| > 0 $, Step 1, 2, or 3 can be applied to the current graph. For each application of Step~1, the size of $K^*$ decreases by 1. For each application of Rule~A in Step~2, the size of $K^*$ decreases by at most 1. Finally, for each application of Step~3, the size of $K^*$ decreases by at most $d+1$ (at most $d$ neighbors of vertex $u$ together with vertex $w$ may not belong to $K^*$ anymore). Since by assumption $G$ is a $d^*$-split graph, initially we have that $|K^*| \geq (d+1)k$, and therefore it follows that $p_1 + p_2 + p_3 \geq k$, concluding the proof.\end{proof}

\begin{claimN}\label{claim:indep-set}
If $|I_h| \geq 2dk$, then $(G,k)$ is a \textsc{Yes}-instance.
\end{claimN}
\begin{proof}
We prove that if $|I_h| \geq 2dk$, then we could iteratively apply Rule~$6^+$ to $(G,k)$ until obtaining an instance $(G',k')$ with $k' \leq 0$, which implies that $G$ is a  \textsc{Yes}-instance.

We apply to $(G,k)$ the following procedure (we stress again that this algorithm is only used for the analysis; as mentioned before we do not modify our input graph at all). We set $i:=1$, $G^1 := G$, $K_h^1 := K_h$,  $I_h^1 := I_h$, and $k_1  := k$. Note that $G^1$ is connected by hypothesis. For the sake of readability, we want to stress that the sets  $K_h^{i}$ and $I_h^{i}$ defined below will {\sl not} correspond to the intersection of $G^i$ with $K_h$ and $I_h$, respectively.  Proceed as follows:

\begin{enumerate}
\item Let $v_i$ be an arbitrary vertex in $K^i_h$, let $N^i := N_{I_h^i}(v_i)$, and let $c_i := |N^i|$. Note that $c_i \geq 2$ and that  $G^i - (\{v_i\} \cup N^i)$ is connected.
\item Apply Rule~$6^+$ to $\{v_i\} \cup N^i$, and let $G^{i+1} := G^i - (\{v_i\} \cup N^i)$.
\item Let $D_K^i$ be the set of vertices in $K_h^i$ having at most one neighbor in $I_h^i \setminus N^i$, and let $D_I^i := N_{I_h^i  \setminus N^i} (D_K^i)$; see Fig.~\ref{fig-linear-kernel} for an example of these sets for $d=3$. Since every vertex in $K_h^i$ has at least two neighbors in $I_h^i$ and each vertex in $I_h^i$ has degree at most $d$, it follows that $|D_K^i| \leq (d-1)|N^i| = (d-1)c_i$. On the other hand, by definition it holds that $|D_I^i| \leq |D_K^i| \leq (d-1)c_i$.
        \item Let $K_h^{i+1}:= K_h^{i} \setminus (\{v_i\} \cup D_K^i \cup N_{K_h^{i}}(D_I^i))$,  $I_h^{i+1}:= I_h^{i} \setminus (N^i \cup D_I^i )$, and $k_{i+1} := k_i - c_i +1$. Note that $ |I_h^{i}| - |I_h^{i+1}| = |N^i| + |D_I^i| \leq c_i + (d-1)c_i = dc_i$, and that by construction, each vertex in $K_h^{i+1}$ has at least two neighbors in $I_h^{i+1}$.
\item If $|I_h^{i+1}| > 0$, update $i \leftarrow i+1$ and go back to Step 1. Otherwise, stop the procedure.
\end{enumerate}

\begin{figure}[h]
\begin{center}
 \includegraphics[width=.5\textwidth]{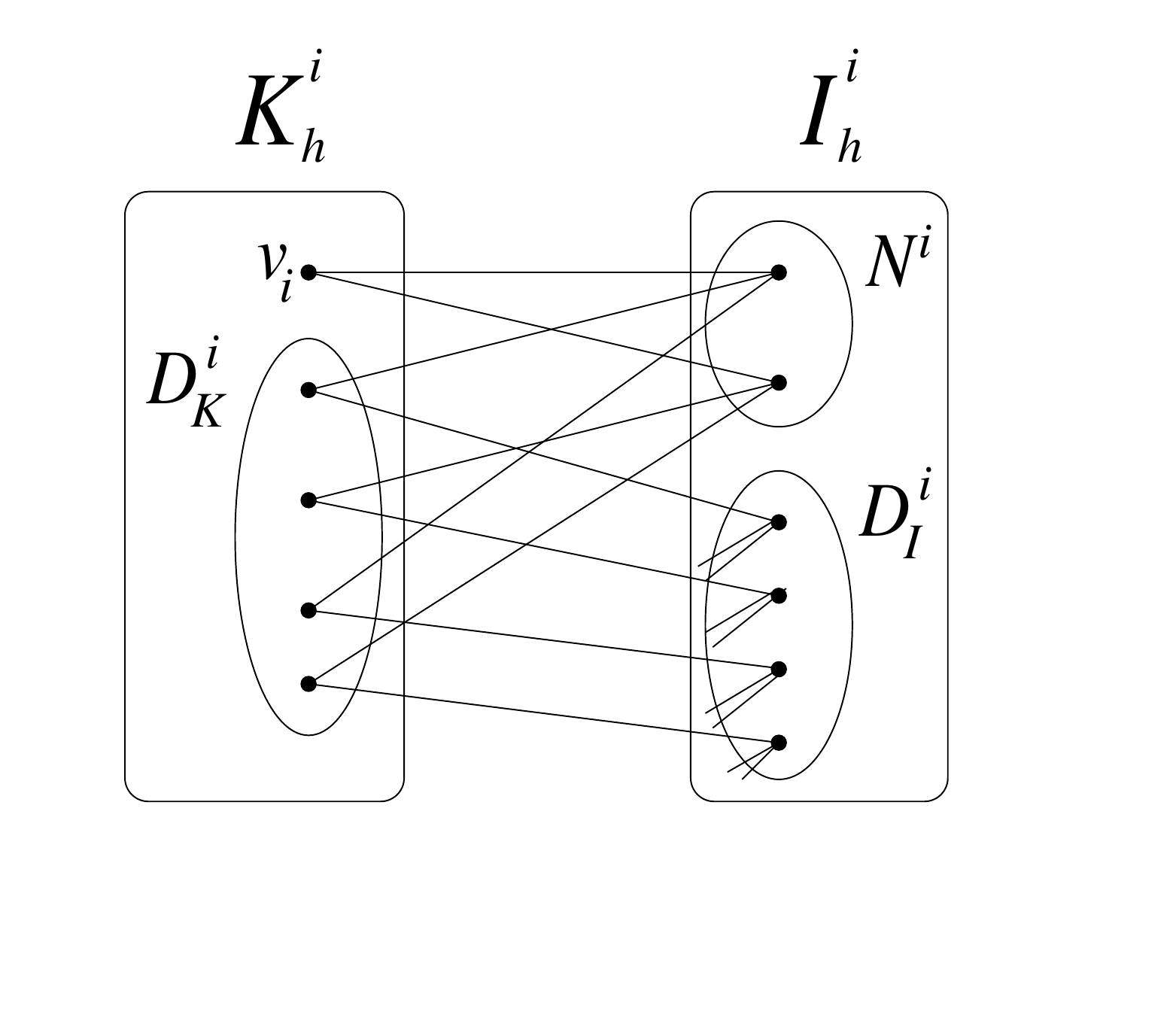}\vspace{-.6cm}
\caption{Example for $d=3$ of the sets defined in the proof of Theorem~\ref{thm:linear-d-split}.} \label{fig-linear-kernel}
\end{center}
\end{figure}

Assume that the above procedure has been applied $p \geq 1$ times, and let $k' := k_{p+1}$. Note that $k' = k - \sum_{i=1}^p (c_i -1)$. Since $k'$ is the parameter obtained from $k$ by iteratively applying Rule~$6^+$ in $G$ to the sets described in Step~2 above, our objective is to prove that $k' \leq 0$.

By the condition in Step~5 above, necessarily $|I_h^{p+1}| = 0$. Since for every $1 \leq i \leq p$ we have that, as discussed in Step~4 above, $ |I_h^{i}| - |I_h^{i+1}| \leq  dc_i$, using the hypothesis that $|I_h^1| = |I_h|\geq 2dk$ it follows that
$$
2dk \leq |I_h^1| = (|I_h^1| - |I_h^2|) + (|I_h^2|-|I_h^3|) + \ldots +  (|I_h^p| - |I_h^{p+1}|) = \sum_{i=1}^p (|I_h^i| - |I_h^{i+1}|) \leq d \cdot \sum_{i=1}^p c_i,
$$
and therefore it follows that $\sum_{i=1}^p c_i \geq 2k$. As for every $1 \leq i \leq p$ it holds that $c_i \geq 2$ (see Step~1 above) and the function $x \mapsto x-1$ is greater than or equal to the function $x \mapsto x/2$ for $x \geq 2$, we have that $\sum_{i=1}^p (c_i - 1) \geq \sum_{i=1}^p c_i / 2 \geq k$. Finally,
$$
k' = k - \sum_{i=1}^p (c_i -1) \leq  k - k = 0,
$$
and therefore we can conclude that  $(G,k)$ is a \textsc{Yes}-instance, as claimed.
\end{proof}

Assume now that the hypothesis of Claim~\ref{claim:clique} and Claim~\ref{claim:indep-set} are both false, that is, that $|K| < (d+1)k$ and $|I_h| < 2dk$. Then we have that
$$
|V| = |K| + |I| = |K| + |I_h| + |I_s| < (d+1)k + 2dk + (d+1)k < 4(d+1)k.
$$

Thus, if the above inequality is not true, that is, if $|V| \geq 4(d+1)k$, then necessarily at least one of the hypothesis of Claim~\ref{claim:clique} or Claim~\ref{claim:indep-set} is true, and in any case we can safely conclude that $(G,k)$ is a \textsc{Yes}-instance. Therefore, our linear kernel is extremely simple: if $|V| \geq 4(d+1)k$, we report that $(G,k)$ is a \textsc{Yes}-instance, and otherwise we have that $|V| < 4(d+1)k$, as desired.\end{proof}

\newpage

\section{Conclusions}
\label{sec:conclusions}

%
%

In this article we presented a quadratic kernel for  {\sc{Signed Max Cut ATLB}} when the input graph is an $(r, \ell)$-graph, for any fixed positive value of $r$ and $\ell$. Etscheid and Mnich~\cite{EtscheidM16} recently presented, among other results, a {\sl linear} kernel for this problem on {\sl general} graphs, improving the results of the current paper, as well as those of Crowston et al.~\cite{crowston2012max}. Like us, the linear kernel of Etscheid and Mnich~\cite{EtscheidM16} also relies on a number of one-way reduction rules and results given by Crowston et al.~\cite{crowston2012max}. The key idea is to use two new two-way reduction rules (similar to those of Crowston et al.~\cite{crowston2012max}) that, roughly speaking, maintain the connectivity among the blocks of $G-S$. The analysis of the kernel size in~\cite{EtscheidM16} is considerably more involved that ours and entails, in particular, an elaborated ``discharging'' argument on the number of vertices removed by the rules.

Concerning {\sf{FPT}}-algorithms, it may be possible that {\sc{Signed Max Cut ATLB}} can be solved in {\sl subexponential} time on $(r,\ell)$-graphs or, at least, on split graphs. We leave this question as an open problem.

\acknowledgements
\label{sec:ack}
We would like to thank the anonymous referees for helpful remarks that improved the presentation of the manuscript.

\bibliographystyle{abbrv}	
\bibliography{maxbalancsg2-Ignasi}
\label{sec:biblio}

\end{document}